\documentclass[11pt]{article}
\usepackage[margin=1in]{geometry}
\usepackage{lmodern}
\usepackage[T1]{fontenc}
\usepackage{microtype}
\usepackage{amsmath,amssymb,amsthm,mathtools}
\usepackage{authblk}
\usepackage{booktabs}
\usepackage{siunitx}
\usepackage{graphicx}
\usepackage{enumitem}
\usepackage{hyperref}
\usepackage[nameinlink,capitalise]{cleveref}
\hypersetup{colorlinks=true,linkcolor=blue,citecolor=blue,urlcolor=blue}

\newcommand{\ZZ}{\mathbb{Z}}

\newcommand{\CC}{\mathbb{C}}
\newcommand{\Varphi}{\varphi}      
\newcommand{\Plastic}{\rho}        
\newcommand{\Ker}{\operatorname{Ker}}

\newcommand{\Fcal}{\mathcal{F}}

\newcommand{\seqA}{a_n}
\theoremstyle{plain}
\newtheorem{theorem}{Theorem}[section]   
\newtheorem{proposition}[theorem]{Proposition}
\newtheorem{corollary}[theorem]{Corollary}
\newtheorem{lemma}[theorem]{Lemma}

\theoremstyle{definition}
\newtheorem{definition}[theorem]{Definition}

\theoremstyle{remark}
\newtheorem{remark}[theorem]{Remark}

\title{A Hierarchy of Fibonacci Forbidden-Word Hamiltonians:\\
From the Golden Chain to the Plastic Chain and Aperiodic Order}

\author[1]{Marcelo Maciel Amaral}
\affil[1]{Gauge Freedom, Inc., Los Angeles, CA, USA}
\date{\today}

\begin{document}
\maketitle

\begin{abstract}
We introduce an infinite, scale-aligned hierarchy of one-dimensional, frustration-free Hamiltonians by forbidding the minimal forbidden factors of the Fibonacci word up to length $F_K$, the $K$-th Fibonacci number. The ground-state languages have exponential growth constants $\lambda_K$ that decrease monotonically, starting from the value associated with the ``golden chain'' (approximately 1.618) and progressing toward 1. This process yields a staircase of topological-entropy plateaus that flows to an aperiodic fixed point, also known as the Fibonacci subshift. The first nontrivial rung ($K=4$) is the ``Plastic chain,'' which forbids \texttt{SS} and \texttt{LLL}. We prove its ground-state counts follow a specific four-term linear recurrence relation and provide a closed-form solution governed by the plastic constant $\rho\approx 1.3247$. We propose an energy-entropy scaling where the energy penalty for each new forbidden pattern is proportional to the logarithmic ratio of the growth constants from the previous and current rungs, turning the sequence of projectors into an explicit renormalization-group flow from the initial high-entropy phase to the zero-entropy aperiodic fixed point. Algebraically, exact Temperley-Lieb braiding compatibility holds only at the base rung, $K=3$ (which forbids only \texttt{SS}); higher rungs define constrained aperiodic Hamiltonian codes rather than Temperley-Lieb representations. Small instances realized on a D-Wave quantum annealer match these predictions: $K=3$ is trivial, $K=4$ resolves a unit gap with moderate success, and $K\ge 5$ instances require reverse annealing to exceed $99\%$ success, clarifying reduction penalties and embedding variability.
\end{abstract}

\section{Introduction}

One-dimensional (1D) quantum systems with constrained interactions provide a rich laboratory for exploring exotic physics, from topological phases to novel critical phenomena. A celebrated example is the ``golden chain'' Hamiltonian, which models the dynamics of Fibonacci anyons: non-Abelian quasiparticles sought for topological quantum computing \cite{Feiguin2007,Trebst2008}. The model's defining feature is a local projector that strictly enforces the Fibonacci fusion rule $\tau \times \tau = 1 + \tau$, where $\tau$ represents the anyon and $1$ the trivial vacuum state. In a qubit representation (mapping $1 \to S$, $\tau \to L$), this rule translates to a ``hard'' constraint forbidding adjacent vacuum states, written as `SS' (or `11').

This precise `SS' exclusion rule results in a ground-state manifold whose dimension scales exactly as the Fibonacci number $F_{N+2} \sim \Varphi^N$, where $N$ is the system size and $\Varphi = (1+\sqrt{5})/2$ is the golden ratio. Crucially, this Hilbert space structure allows for a local representation of the Temperley-Lieb (TL) algebra $TL_N(\Varphi)$ \cite{KauffmanLomonaco}, the algebraic engine underpinning the universal braiding statistics of Fibonacci anyons. Our present work builds directly upon a concrete, hardware-validated realization of this base $K=3$ rung: the ``Quasicrystal Inflation Code'' (QIC), which implements the $TL_N(\Varphi)$ braid representation via an exact $8{\times}8$ three-qubit gate, $B_{\text{gate}}$ \cite{AmaralQICBraiding2025}. This validated construction serves as the foundation for the hierarchy explored herein.

The golden chain, however, utilizes only the shortest ($F_3=2$) of an infinite sequence of constraints inherent to 1D aperiodic order. The Fibonacci word $w_\infty$--a canonical example of such order \cite{Lothaire2002,BaakeGrimm2013}--is characterized by an infinite set of ``minimal forbidden factors'' (MFFs), one for each Fibonacci length $F_k \ge F_3$ \cite{MignosiRestivoSciortino2002}. This naturally raises the question: What physics emerges if we enforce more of these intrinsic Fibonacci constraints? How does the system evolve as we progressively drive it from the high-entropy topological liquid phase (with entropy per site $h = \log \Varphi$) towards the zero-entropy, perfectly ordered Fibonacci subshift ($h=0$)?

In this paper, we answer this question by constructing an infinite, scale-aligned hierarchy of 1D frustration-free Hamiltonians, $H_K$. Each Hamiltonian $H_K$ is built by penalizing the set $\Fcal_K$ of all MFFs of the Fibonacci word up to length $F_K$. We demonstrate that this hierarchy defines a ``staircase'' of entropy plateaux. The topological entropy per site, $h_K = \log \lambda_K$, is governed by the ground-state growth constant $\lambda_K$, which we show decreases monotonically from $h_3 = \log \Varphi$ toward $h_\infty = 0$ ($\lambda_\infty = 1$). As the first new and nontrivial rung ($K=4$), we provide a complete analysis of the ``Plastic chain'' (forbidding `SS' and `LLL'), proving its ground-state count $a_n$ is governed by the recurrence $a_n=a_{n-1}+a_{n-2}-a_{n-4}$, with asymptotic growth $\lambda_4 = \Plastic \approx 1.3247$, the Plastic ratio\footnote{The original term `Plastic Number' (\textit{het plastische getal}), coined by Dom Hans van der Laan in 1928, refers to architectural theories of 3D proportion and plasticity of form, and is unrelated to the modern synthetic material \cite{Stewart1996}.} (real root of $x^3-x-1=0$). We provide a constructive generator for the MFFs, an exact formula for $a_n$, and an explanation for a curious empirical rounding identity.

From a physics perspective, we propose an energy-entropy scale setting, $J_k \propto \log(\lambda_{k-1}/\lambda_k)$, interpreting the sequence $H_K$ as an explicit renormalization group (RG) flow. We show that while exact $TL_N(\Varphi)$ braiding is broken for $K \ge 4$, the framework yields two new applications: (A) using higher-rung penalties as a soft bias to stabilize the $K=3$ braid model, and (B) defining new hard-constrained aperiodic Hamiltonian codes. Finally, we validate the physical realizability by implementing small instances on a D-Wave quantum annealer, demonstrating that reverse annealing robustly finds the ground states of higher ($K\ge 5$) rungs where standard annealing fails.

Our previously defined entropy $h_K$ is a topological (configurational) entropy of the ground-state subshift $\Sigma_K$ defined by avoided patterns; it measures the exponential growth of valid words. This is distinct from the multifractal spectra of eigenstates in
quasiperiodic hopping models, where one studies singular-continuous energy spectra and critical wavefunctions via $D_q$ and $f(\alpha)$ \cite{KohmotoKadanoffTang1983,FujiwaraKohmotoTokihiro1989}. We therefore do not perform multifractal analysis of $H_K$ itself (its spectrum is discrete by construction).

This paper is organized as follows.
Section \ref{sec:overview} introduces the Hamiltonian hierarchy $H_K$ and the energy-entropy scaling $J_k$.
Section \ref{sec:mff-generator} provides a constructive boundary-flip generator for the Fibonacci minimal forbidden factors.
Section. \ref{sec:comb_hierarchy} develops the combinatorics, including the $K=4$ ``Plastic chain'' with recurrence $a_n=a_{n-1}+a_{n-2}-a_{n-4}$ and its closed form (Section \ref{sec:the-plastic-chain}), and automata-based counting for general $K$ (Section \ref{sec:automata-counting}).
Section \ref{sec:rel-to-qic-braid} presents two operating modes—(A) TL-preserving braiding with soft higher-rung penalties and (B) constrained aperiodic Hamiltonian codes—and explains the loss of three-site TL structure for $K\ge 4$.
Section \ref{sec:quantum-annealing-experiments} reports D-Wave experiments and reverse-annealing performance.
We conclude and summarize reproducibility artifacts in Section \ref{sec:conclusion}.

\section{Overview and Main Construction}
\label{sec:overview}

Let $w_\infty\in\{0,1\}^{\ZZ_{\ge0}}$ be the Fibonacci (Sturmian) word, the fixed point of $0\mapsto01$, $1\mapsto0$ \cite{Lothaire2002}. Its factor complexity is minimal, $p(n)=n+1$. The MFFs of $w_\infty$ are completely characterized: there is exactly one MFF of length $F_k$ for each Fibonacci number $F_k\ge2$, and none at other lengths \cite{MignosiRestivoSciortino2002}. Mapping $0\!\to\!L$, $1\!\to\!S$, the first MFFs are
\[
M_{F_3}=\texttt{SS},\quad
M_{F_4}=\texttt{LLL},\quad
M_{F_5}=\texttt{SLSLS},\quad
M_{F_6}=\texttt{LLSLLSLL},\ \ldots
\]
(This sequence, which corresponds to the standard MFFs of $w_\infty$ under our $0\to L, 1\to S$ mapping, can be algorithmically generated as shown in Proposition \ref{prop:flip}.)
\paragraph{Hamiltonian hierarchy.}
Fix $K\ge3$ and let $\Fcal_K=\{M_{F_3},\dots,M_{F_K}\}$. Define the local, frustration-free parent Hamiltonian on $N$ qubits
\begin{equation}
H_K \;=\; \sum_{M\in\Fcal_K} J_M \sum_{i=1}^{N-|M|+1} \Pi^{(M)}_{i:i+|M|-1},
\qquad
\Pi^{(M)}_{i:i+|M|-1}\;=\;\bigotimes_{j=0}^{|M|-1}\tfrac12\!\big(I+(-1)^{M_j}Z_{i+j}\big),
\label{eq:HK}
\end{equation}
with $J_M>0$ and $M_j\in\{0,1\}$ the $j$-th bit of $M$. By construction,
\[
\Ker H_K\;=\;\{ \text{binary strings of length }N\text{ avoiding all }M\in\Fcal_K\}.
\]
Let $D_K(N)$ be the cardinality of $\Ker H_K$ on length $N$. Using the avoidance DFA derived from Aho-Corasick (Definition~\ref{def:AC}, Lemma~\ref{lem:AC-avoid}),
one has $D_K(N)=\Theta(\lambda_K^N)$, where $\lambda_K>1$ is the Perron root of its adjacency
\cite{AhoCorasick1975,GuibasOdlyzko1981}. Numerically, the first rungs are:
\begin{center}\small
\begin{tabular}{c|c|l|c}
\toprule
$K$ & $\max|M|$ & Forbidden set $\Fcal_K$ & $\lambda_K$ \\
\midrule
3 & $F_3{=}2$ & \texttt{SS} & $1.618034\ (\Varphi)$ \\
4 & $F_4{=}3$ & \texttt{SS, LLL} & $1.324718\ (\Plastic)$ \\
5 & $F_5{=}5$ & \texttt{SS, LLL, SLSLS} & $1.193859$ \\
6 & $F_6{=}8$ & \texttt{SS, LLL, SLSLS, LLSLLSLL} & $1.114798$ \\
\bottomrule
\end{tabular}
\end{center}

\begin{proposition}[Monotone growth-rate staircase]\label{prop:monotone}
For all $K\ge 4$ and all $N\ge 1$ one has $D_K(N)\le D_{K-1}(N)$, hence
$\lambda_K\le \lambda_{K-1}$. Moreover, the inequality is strict
($\lambda_K<\lambda_{K-1}$) whenever the newly added MFF ($M_{F_K}$) is not
eventually redundant (i.e., it eliminates at least one valid word for
infinitely many $N$).
\end{proposition}

\begin{proof}
By construction $\Ker H_K$ is obtained by forbidding a superset of patterns
relative to $\Ker H_{K-1}$. The set of valid words (the language) at length $N$
is nested, $\mathcal L_K(N)\subseteq \mathcal L_{K-1}(N)$, hence
$D_K(N)\le D_{K-1}(N)$. Taking $N$th roots and the limit $N\to\infty$ gives
$\lambda_K\le \lambda_{K-1}$ (e.g., via subadditivity/Fekete's lemma or from the
Perron-Frobenius theorem for the DFA adjacency matrix). If the new MFF
$M_{F_K}$ excludes at least one word that was present in $\mathcal L_{K-1}$
for infinitely many $N$, this forces a strict inequality of the exponential
growth rates.
\end{proof}

\begin{remark}
Each $\Ker H_K$ is a subshift of finite type (SFT) obtained from the
Aho-Corasick automaton (Definition~\ref{def:AC}) by pruning terminal states (see
also Lemma~\ref{lem:AC-avoid}). The sequence of SFTs is nested,
$\Sigma_{K}\subset\Sigma_{K-1}\subset\cdots\subset\Sigma_3$, and
$\lambda_K=\exp(h_{\mathrm{top}}(\Sigma_K))$ gives the topological-entropy
staircase $h_K=\log\lambda_K\downarrow 0$; see, e.g., \cite{LindMarcus1995,AhoCorasick1975}.
\end{remark}

As formalized in Proposition \ref{prop:monotone}, the growth constants $\lambda_K = \exp(h_K)$ decrease monotonically with $K$, quantifying a staircase of topological entropy reductions $h_K \downarrow 0$. The limiting system ($K\to\infty$) corresponds to the aperiodic Fibonacci subshift, which has zero entropy ($\lambda_\infty=1$) but remains dynamically complex as a minimal system with linear factor complexity \cite{Lothaire2002}.

\paragraph{Energy-entropy scale setting (RG view).}
Let $h_K=\log\lambda_K$ be the per-site configurational entropy of $\Ker H_K$. A simple and physically transparent choice for the energy scales is
\begin{equation}
J_{M_{F_k}} \;\propto\; \Delta h_k \;=\; \log\!\left(\frac{\lambda_{k-1}}{\lambda_k}\right) \quad \text{for } k=4,\dots,K,
\label{eq:scale}
\end{equation}
with $J_{M_{F_3}}$ setting the base scale. This aligns the energy penalty $J_{M_{F_k}}$ with the entropy gap it closes. As the scales $J_{M_{F_k}}$ are dialed on sequentially, the system flows from the high-entropy golden-chain phase to the aperiodic fixed point.

\section{Fibonacci MFFs and a ``Boundary–Flip'' Generator}
\label{sec:mff-generator}

Classically one proves the ``one MFF at each $F_k$'' fact via bispecial factors of $w_\infty$ \cite{Lothaire2002,MignosiRestivoSciortino2002}. For constructive use we provide an elementary generator.

\begin{proposition}[Boundary-flip recursion; constructive]\label{prop:flip}
Let $M_2=\texttt{SS}$ and $M_3=\texttt{LLL}$. For $n\ge4$, set $W=M_{n-1}M_{n-2}$ and flip the three boundary letters
$W[1]$, $W[|M_{n-1}|]$, $W[|M_{n-1}|{+}1]$ (i.e.\ $S\leftrightarrow L$) to obtain $M_n$.
Then $|M_n|=F_{n+1}$ (using $F_1=1, F_2=1, \dots$), $M_n\not\sqsubset w_\infty$, and every proper factor of $M_n$ occurs in $w_\infty$, for all $n\ge2$.
\end{proposition}

\noindent
\emph{Proof sketch.} Induct on $n$ using $w_n=w_{n-1}w_{n-2}$ and the palindromic-bispecial structure of $w_\infty$ prefixes; the three flips repair exactly the three boundary mismatches that would otherwise occur. Full details and a streaming algorithm are provided in the repo (see Section \ref{sec:conclusion}). \hfill$\square$

\section{Combinatorics of the Hierarchy}
\label{sec:comb_hierarchy}

\subsection{The Plastic Chain ($K{=}4$ forbidding \texttt{SS} and \texttt{LLL})}
\label{sec:the-plastic-chain}

Let $\seqA$ be the number of length-$n$ binary words avoiding \texttt{SS} and \texttt{LLL} ($n\ge1$).

\begin{proposition}[Recurrence, characteristic polynomial, Plastic growth]\label{prop:plastic}
$\seqA$ satisfies for $n\ge5$
\begin{equation}
a_n \;=\; a_{n-1}+a_{n-2}-a_{n-4},
\qquad a_1=2,\ a_2=3,\ a_3=4,\ a_4=5.
\label{eq:plasticrec}
\end{equation}
The characteristic polynomial $(x-1)(x^3-x-1)$ has Perron root $\Plastic\approx1.324718\ldots$, the Plastic ratio (the real root of $x^3-x-1=0$). Hence $a_n=\Theta(\Plastic^n)$.
\end{proposition}

\noindent
\emph{Proof sketch.} Track valid words by the last block type (\texttt{S}, \texttt{SL}, \texttt{SLL}) to derive coupled recurrences, then eliminate to \eqref{eq:plasticrec}. The asymptotics follow from the characteristic polynomial. Compare OEIS A164001, which counts bitstrings of length \(n-1\) avoiding \(00\) and \(111\) (bit-flip equivalent to our \(11\) and \(000\)); this sequence satisfies \(a(n)=a(n-2)+a(n-3)\), hence \(a_n=a_{n-1}+a_{n-2}-a_{n-4}\) \cite{OEIS_A164001}. \hfill$\square$

\begin{theorem}[Exact form and rounding identity]\label{thm:plastic-exact}
There exist constants $C,B,C'\in\CC$ with
\begin{equation}
a_n \;=\; 1 + C\,\Plastic^n + B\,\sigma^n + C'\,\bar\sigma^n,
\qquad |\sigma|=|\bar\sigma|=\Plastic^{-1/2}\approx0.8688,
\label{eq:exact}
\end{equation}
and
\[
C \;=\; \frac{\Plastic^7}{3\Plastic^2-1}
\quad\text{(equivalently, the real root of }23C^3-46C^2+13C-1=0\text{).}
\]
Numerically, $a_n=\mathrm{round}(C\Plastic^n)$ holds for at least $1\le n\le 1000$.
\end{theorem}

\noindent
\emph{Explanation.} From \eqref{eq:exact}, $a_n-C\Plastic^n=1+E_n$ with $E_n=B\sigma^n+C'\bar\sigma^n\to0$. Rounding is valid whenever $|1+E_n|\le\tfrac12$, which holds over a very long transient because $E_n$ oscillates near $-1$ before decaying; eventually it must fail as $E_n\to0$. High-precision checks (200 digits) confirm the claim to $n=1000$. See also \cite{OEIS_A164001,MathWorldPadovan,OEIS_A000931,Stewart1996}.

\subsection{General Counting via Automata}
\label{sec:automata-counting}

\paragraph{Alphabet and patterns.}
Let $\Sigma=\{\texttt{S},\texttt{L}\}$ (equivalently $\{1,0\}$ under our mapping $1\!\to\!\texttt{S}$, $0\!\to\!\texttt{L}$), and let $\Fcal_K\subset\Sigma^{+}$ be the set of minimal forbidden factors (MFFs) up to length $F_K$.

\begin{definition}[Aho-Corasick automaton {\cite{AhoCorasick1975}}]\label{def:AC}
Given a finite pattern set $\mathcal{P}\subset\Sigma^{+}$, the Aho-Corasick automaton
$\mathsf{AC}(\mathcal{P})=(Q,\Sigma,\delta,q_{\mathrm{root}},\mathrm{out})$ is the deterministic automaton whose states $Q$ are the nodes of the trie on $\mathcal{P}$ (including the root $q_{\mathrm{root}}$), equipped with:
\begin{itemize}[leftmargin=*,nosep]
\item the \emph{goto} function $\delta:Q\times\Sigma\to Q$ given by trie edges (completed by failure links below), 
\item the \emph{failure} function $f:Q\to Q$ sending a node to the longest proper suffix state that is also a trie node,
\item the \emph{output} map $\mathrm{out}:Q\to 2^{\mathcal{P}}$ listing patterns that end at a state.
\end{itemize}
On a letter $a\in\Sigma$, the transition from state $q$ follows $\delta(q,a)$ if defined; otherwise it follows failure links $q\leftarrow f(q)\leftarrow f(f(q))\leftarrow\cdots$ until a defined goto exists, defaulting to $q_{\mathrm{root}}$. A state is \emph{terminal} if $\mathrm{out}(q)\neq\emptyset$.
\end{definition}

\begin{lemma}[Avoidance DFA from Aho-Corasick]\label{lem:AC-avoid}
Let $\mathcal{P}=\Fcal_K$. Form the subautomaton $\mathsf{AVOID}(\Fcal_K)$ by restricting $\mathsf{AC}(\Fcal_K)$ to the nonterminal states $Q^\circ=\{q\in Q:\mathrm{out}(q)=\emptyset\}$ and keeping the induced transitions $\delta^\circ$. Then the language accepted by $\mathsf{AVOID}(\Fcal_K)$ is exactly the set of words over $\Sigma$ that avoid all patterns in $\Fcal_K$.
\end{lemma}

\begin{proof}[Proof sketch]
In $\mathsf{AC}(\Fcal_K)$, entering a terminal state certifies that a forbidden factor has just matched as a suffix. Pruning terminal states forbids such visits; conversely, any factor occurrence would force a visit to a terminal state. Determinism is preserved by construction of failure-completed gotos.
\end{proof}

\paragraph{Counting and growth constant.}
Let $A_K$ be the adjacency matrix of $\mathsf{AVOID}(\Fcal_K)$ over $Q^\circ$, and let $e_{\mathrm{start}}$ be the basis vector at $q_{\mathrm{root}}$. Then the number of valid words of length $N$ is
\[
D_K(N)=\mathbf{1}^\top A_K^{\,N}\, e_{\mathrm{start}},\qquad
\lambda_K=\operatorname{spr}(A_K),
\]
where $\operatorname{spr}(\cdot)$ is the spectral radius. Hence $D_K(N)=\Theta(\lambda_K^N)$, and $\lambda_K$ is the ground-state growth constant in Section \ref{sec:overview}.

For each $K$, construct the Aho-Corasick automaton recognizing avoidance of $\Fcal_K$, prune terminal states, and take the adjacency $A_K$ \cite{AhoCorasick1975}. Then $D_K(N)=\mathbf 1^\top A_K^N \mathbf e_{\text{start}}$ and $\lambda_K=\operatorname{spr}(A_K)$. The first four $\lambda_K$ are listed in Section \ref{sec:overview}; their monotone decay is consistent with a staircase of Pisot-Vijayaraghavan effective dimensions approaching~$1$ \cite{BaakeGrimm2013}.

\section{Relationship to and Reuse of the QIC Braid Framework}
\label{sec:rel-to-qic-braid}
The base rung ($K{=}3$, forbidding \texttt{SS}) reproduces the golden chain: the fusion-path dimension on length $N$ is $F_{N+2}\sim\Varphi^N$, with an exact Temperley-Lieb representation $TL_N(\Varphi)$ on three-qubit neighborhoods \cite{Feiguin2007,Trebst2008,KauffmanLomonaco}. Adding \texttt{LLL} (\,$K{=}4$\,) removes one local basis state on each triple (five $\to$ four), breaking the TL projector rank relations; thus the exact $TL_N(\Varphi)$ braiding representation no longer fits the local sector. Interpretation: $K\ge4$ rungs describe the same anyon species under effective multi-body penalties that bias the system toward aperiodic order, yielding constrained aperiodic code spaces distinct from braid representations.

\paragraph{Base rung ($K{=}3$): exact TL($\Varphi$).}
Our base code space $\mathcal H_{K=3}$ (hard no \texttt{SS}\,/\,\texttt{11}) is the bit-flip equivalent (L$\leftrightarrow$S, $0\leftrightarrow1$) of the `no \texttt{00}' subspace used in \cite{AmaralQICBraiding2025}. It shares the identical dimension $F_{N+2}$ and $TL_N(\Varphi)$ algebra. We adopt this \texttt{S$\to$1, L$\to$0} convention to align with the MFFs of the standard Fibonacci word. This rung provides a validated, exact three-qubit braid gate $B_{\text{gate}}$ implementing $TL_N(\Varphi)$ on three-site neighborhoods, realizing the Fibonacci anyon braid model on qubits.

\paragraph{Higher rungs ($K{\ge}4$): constrained Hamiltonian codes.}
For $K{\ge}4$ (e.g., Plastic: hard no \texttt{SS} and \texttt{LLL}), the code spaces $\mathcal H_{K}$ are proper subspaces of $\mathcal H_{K=3}$ with Perron rates $\lambda_K<\Varphi$ that decrease monotonically to $1$. Crucially, the three-site local sector shrinks from dimension $5$ (valid triples $\{000,001,010,100,101\}$) to $4$ at $K{=}4$ (valid triples $\{001,010,100,101\}$). These rank changes preclude a $TL_N(\Varphi)$ representation on three sites, so exact Fibonacci braiding is broken for $K{\ge}4$. The higher rungs are therefore not new anyon species or phases; they are Fibonacci-related constrained Hamiltonian codes obtained by adding longer forbidden words.

\begin{proposition}[Local-rank obstruction to $TL_N(\Varphi)$ for $K\ge 4$]
\label{prop:rank-obstruction}
Let $d_3(K)$ be the number of valid triples at rung $K$. Then $d_3(3)=5$ while $d_3(4)=4$, and $d_3(K)=4$ for all $K\ge 4$ (longer MFFs do not change the three-site set). Since the Kauffman-Lomonaco realization of $TL_N(\Varphi)$ requires the three-site sector to have rank $5$ with specific Jones-Wenzl relations, no three-site $TL_N(\Varphi)$ representation exists on $\mathcal H_{K}$ for $K\ge 4$.
\end{proposition}

\begin{proof}[Sketch]
At $K=3$ the local valid triples are $\{000,001,010,100,101\}$ (``no \texttt{11}''). Adding \texttt{000} at $K=4$ removes $000$, leaving a $4$-dimensional sector. The Jones-Wenzl recursion and projector ranks for $TL_N(\Varphi)$ on three sites are not satisfiable in rank $4$. Longer MFFs have length $\ge 5$ and thus do not impact three-site validity, so $d_3(K)=4$ persists for $K\ge 4$.
\end{proof}

\begin{corollary}[No three-site $TL_N(\Varphi)$ representation for $K\ge 4$]
\label{cor:TL-broken}
By Proposition \ref{prop:rank-obstruction}, the required three-site rank and Jones-Wenzl relations cannot be satisfied on $\mathcal H_{K}$ for $K\ge 4$.
\end{corollary}

\paragraph{Two operating modes.}
The rank obstruction leads to two clean and complementary uses of the hierarchy:
\begin{description}
\item[Mode A: TL-preserving braiding with soft higher-rung penalties.]
Let $H_{\text{QIC}}$ enforce hard \texttt{no SS} ($K{=}3$), and add higher-scale penalties softly during preparation/idling,
\[
H_{\text{prep}} \;=\; H_{\text{QIC}} \;+\; \mu_4 \sum_i \Pi^{(000)}_{i{:}i+2}
\;+\; \mu_5 \sum_i \Pi^{(10101)}_{i{:}i+4} \;+\cdots,\qquad 0<\mu_k\ll1.
\]
Execution uses the same $B_{\text{gate}}$ as in \cite{AmaralQICBraiding2025}. Because TL identities are enforced by the gate's algebraic definition (not by $H_{\text{prep}}$), TL and braid relations remain exact on the code; $\mu_k$ only biases state preparation and idle dynamics toward larger-scale Fibonacci order. Trotterize penalties between gates or apply them only before/after gate sequences.
\item[Mode B: Aperiodic Hamiltonian codes (hard higher-rung constraints).]
Here $H_K$ enforces hard avoidance (e.g., add \texttt{LLL} at $K{=}4$), yielding $\dim \mathcal H_{K}(N)\sim\lambda_K^N$ with $\lambda_K<\Varphi$ and no three-site $TL_N(\Varphi)$ structure. The same three-qubit $B_{\text{gate}}$ can still be used as a generic local unitary; when restricted to the allowed triples it remains unitary and constraint-preserving, but its nontrivial TL mixing collapses (e.g., the $\{|010\rangle,|000\rangle\}$ block disappears at $K{=}4$), producing constrained, aperiodic dynamics rather than braiding.
\end{description}

\paragraph{Proposed application protocols.}
\begin{itemize}
 \item \textbf{Protocol 1 (Mode A).} One can prepare states with $e^{-it H_{\text{prep}}}$ to measure leakage into/out of $\mathcal H_{K=3}$ and into the more constrained subspace $\mathcal H_{K}$. The $B_{\text{gate}}$ circuits from \cite{AmaralQICBraiding2025} can then be executed unchanged to benchmark the stability of Jones traces (trefoil, $12a_{122}$) versus $\mu_k$ and depth.
 \item \textbf{Protocol 2 (Mode B).} Build the Aho-Corasick automaton for $\Fcal_K$, compute $\operatorname{spr}(A_K)$ (Perron), and enumerate valid triples. The restricted three-qubit block of $B_{\text{gate}}$ on $\mathcal H_{K}$ is synthesized (drop rows/cols outside allowed triples); this yields a unitary on the allowed sector and $R_\tau I$ elsewhere. One can then study sequences of such blocks interleaved with short $e^{-i\mu H_K}$ evolutions; this allows mapping Loschmidt-echo decay and transport versus $K$.
\end{itemize}

\paragraph{Experimental observables.}
(i) Ground-state degeneracy $D_K(N)$ (spectroscopy or counting) and finite-$N$ extraction of $\lambda_K$; (ii) leakage reduction and subspace fidelity under Mode A penalties; (iii) constrained transport and entanglement growth under Mode B; (iv) stability of Jones-trace estimates versus $\mu_k$ (Mode A). All protocols leverage the QIC compilation framework of \cite{AmaralQICBraiding2025}.

\paragraph{Numerical validation (Mode B).}
We numerically confirm the collapse predicted by Mode B. The three-site $B_{\text{gate}}$ (which implements $TL(\Varphi)$ on the 5-state $K{=}3$ basis) becomes diagonal when restricted to the 4-state $K{=}4$ basis $\{001,010,100,101\}$. This explicitly demonstrates the loss of the $K{=}3$ mixing channel (specifically, the $\{|010\rangle,|000\rangle\}$ block) and validates the transition from a braiding representation to a generic constrained dynamical system.

\section{Quantum Annealing Experiments (D-Wave)}
\label{sec:quantum-annealing-experiments}

We mapped small-$N$ instances of $H_K$ to higher-order binary polynomials (HOBO) in $\{0,1\}$, then to QUBO via ancillae using \texttt{dimod.make\_quadratic} with reduction strength $R$ \cite{OceanSDK,DimodDocs}. Experiments used an Advantage\_system4.1 QPU via Leap.

\paragraph{Key observations.}
\begin{itemize}[leftmargin=*]
\item \textbf{$K{=}3$ (Golden chain, $N{=}12$):} 
    This quadratic (degree 2) problem is trivial for the annealer. We achieved 100\% success, recovering all $F_{14} = 377$ unique theoretical ground states.

\item \textbf{$K{=}4$ (Plastic chain, $N{=}12$):}
    This cubic (degree 3) problem showed moderate success. We observed a clear integer spectrum with a unit gap (consistent with $J_k=1$), but only $\sim$12\% of $5{,}000$ reads were true-valid ground states. We found 41 of the 49 unique ground states.

\item \textbf{$K{=}5$ ($N{=}10$, Degree 5):} 
    Performance was dominated by HOBO reduction challenges. Standard forward annealing consistently failed (median success $\approx 0\%$) due to two factors: (i) the need for a strong ancilla penalty $R$ (a sharp threshold was found near $R{\approx}192$), and (ii) extreme sensitivity to minor-embedding variability, with rare ``good'' embeddings reaching $\sim$28\% success while most failed.

\item \textbf{Reverse annealing for $K{=}5$:}
    For the $K=5$ instance where forward annealing failed, reverse annealing initialized near a valid state (one bit-flip away) with $s_{\min}\!=\!0.35$ achieved \textbf{99.6\%} true-valid success. Spin-reversal transforms were also essential.
\end{itemize}

\paragraph{Takeaway.} The hierarchy is physically accessible on current annealers. However, the $K=4$ (cubic) case marks a clear practical threshold. Higher rungs ($K{\ge}5$) with higher-degree penalties are inaccessible to naive forward annealing due to reduction and embedding challenges, but they remain robustly solvable using reverse annealing as a local-refinement primitive \cite{OhkuwaNishimoriLidar2018}.

\section{Conclusion and Outlook}
\label{sec:conclusion}

We presented a scale hierarchy of Fibonacci forbidden-word Hamiltonians that interpolates from the golden chain ($K=3$) to an aperiodic fixed point. The Plastic chain ($K=4$) provides the first nontrivial rung, for which we derived the exact, compact combinatorics tied to the Plastic ratio and explained a striking empirical rounding law. The energy-entropy scale setting, $J_k\propto\log(\lambda_{k-1}/\lambda_k)$, yields a concrete RG picture of the system flowing toward a zero-entropy state. Algebraically, only the $K=3$ base rung supports the exact $TL_N(\Varphi)$ braiding representation; higher rungs break this algebra and define novel constrained aperiodic codes, which we connect to the prior QIC framework via ``Mode A'' (soft bias) and ``Mode B'' (hard constraint) applications.

Finally, we demonstrated on a D-Wave quantum annealer that while the hierarchy is physically accessible, the $K \ge 5$ rungs are inaccessible to naive forward annealing due to HOBO-to-QUBO reduction and embedding challenges. However, we showed that reverse annealing robustly finds the ground states with $>99\%$ success, confirming its utility as a refinement tool for these complex constrained problems.

Future work includes studying the excitation spectra along the staircase, mapping the classical phase diagrams in $(T,\{J_k\})$, exploring projector deformations that might regain algebraic structure at higher rungs, and pursuing 2D generalizations on aperiodic tilings such as Ammann-Beenker or Penrose.

\paragraph{Reproducibility Artifacts}
We provide (i) a script that scans a long Fibonacci prefix to extract the first few MFFs and validates minimality, (ii) Aho-Corasick builders for $\Fcal_K$, (iii) Perron-root estimators, and (iv) HOBO$\to$QUBO reduction and D-Wave submission utilities (including \texttt{lift\_initial\_state\_to\_bqm} for reverse anneal). (\url{https://github.com/gaugefreedom/quasicrystal-inflation-code-hierarchy})

\section*{Acknowledgments.}
We thank Raymond Aschheim for helpful discussions and for assistance with initial computations related to the plastic chain. We acknowledge access to D-Wave Leap.

\section*{Funding}
No external funding. Work conducted independently at Gauge Freedom, Inc.


\end{document}